\documentclass[10pt]{article} 
\usepackage[T1]{fontenc}

\usepackage{geometry} 
\usepackage{helvet, latexsym, amsmath, amsfonts, amssymb, amsthm, graphicx, color, verbatim, alltt, fancyvrb, xcolor, animate, bm, url}
\usepackage{dsfont}
\usepackage{authblk}

\newtheorem{thm}{Theorem}

\def\f{\mathbf f}

\def\H{\mathbf H}

\def\M{\mathbf M}

\def\X{\mathcal X}

\def\1{\mathbf 1}
\def\0{\mathbf 0}

\def\u{\mathbf u}

\def\f{\mathbf f}

\def\NN{\mathbb N}

\def\RR{\mathbb R}

\newcommand{\rank}{\mathrm{rank}}

\title{Optimal Exact Designs of Multiresponse Experiments\\under Linear and Sparsity Constraints}
\author[1]{Lenka Filov\'a}
\author[1]{P\'al Somogyi}
\author[1]{Radoslav Harman}
\affil[1]{Faculty of Mathematics, Physics and Informatics\\
Comenius University, Bratislava, Slovakia}
\date{}

\begin{document}

\maketitle

\begin{abstract}
We propose a computational approach to constructing exact designs on finite design spaces that are optimal for multiresponse regression experiments under a combination of the standard linear and specific 'sparsity' constraints. The linear constraints address, for example, limits on multiple resource consumption and the problem of optimal design augmentation, while the sparsity constraints control the set of distinct trial conditions utilized by the design. The key idea is to construct an artificial optimal design problem that can be solved using any existing mathematical programming technique for univariate-response optimal designs under pure linear constraints. The solution to this artificial problem can then be directly converted into an optimal design for the primary multivariate-response setting with combined linear and sparsity constraints. We demonstrate the utility and flexibility of the approach through dose-response experiments with constraints on safety, efficacy, and cost, where cost also depends on the number of distinct doses used.
\end{abstract}

\\

{\bf Funding:} The research was supported by the Collegium Talentum Programme of Hungary and the Slovak Scientific Grant Agency (grant VEGA 1/0480/26).

\section{Introduction}\label{sec:introduction}

Optimal design of statistical experiments lies at the intersection of optimization and statistics, focusing on developing efficient data acquisition methods (e.g., \cite{Pazman}, \cite{Puk}, \cite{atkinson}, \cite{FedorovHackl}, \cite{jesus}). It aims to select an experiment that, under various constraints, minimizes the uncertainty in parameter estimates of an underlying statistical model, leading to more precise and reliable inference. Optimal experimental design has found applications in industry, clinical research, and many other fields (e.g., \cite{GJ}, \cite{BW}, \cite{Ting}).

We view an experiment as a collection of individual 'runs' or 'trials', each resulting in an observation. Let $w$ be an object representing the experiment, and let $\Xi$ be the set of all permissible designs. On $\Xi$, define a real-valued function $\phi$, called the 'optimality criterion', which serves as a measure of design quality. Note that for a design $w$, the value $\phi(w)$ typically also depends on an underlying statistical model that describes the assumed stochastic distribution of observations under $w$. Then, $w^*$ is an optimal design if it maximizes $\phi$ over $\Xi$, that is, $w^* \in \mathrm{argmax}\{\phi(w) : w \in \Xi\}$.

Each feasible design $w \in \Xi$ is determined by a 'design space' $\X$, also referred to as an 'experimental domain'. The elements of $\X$, called 'candidate design points', represent the conditions under which individual trials can be performed. These are often the levels of a design factor or all combinations of the levels of several design factors.

Extensive literature on optimal experimental design focuses on so-called 'approximate' designs,\footnote{Approximate designs are sometimes also called 'continuous' designs or 'infinite-sample' designs.} formally finitely-supported measures, often defined on continuous design spaces $\X$. The approximate designs prescribe the real-valued proportions of trials performed at individual design points. While approximate designs offer unquestionable theoretical and computational advantages, the ultimate goal in practical experimental design is to construct an 'exact' design, which determines the integer replication numbers at each design point (see, e.g., \cite{Puk}, Chapter 12, \cite{atkinson}, Section 9.1, \cite{jesus}, Section 2.2).

Similarly, a continuous design space is typically only an idealized approximation of reality, as it assumes that trial conditions, such as factor levels, can be selected with infinite precision; in practice, trial conditions are chosen from a design space $\X$ with a finite, albeit possibly large, number $n$ of candidate points. Importantly, in many applications, a finite design space is not only practical but also the most natural, or even the only meaningful, choice for defining the experimental domain. This applies, for instance, to screening designs, block designs, weighing designs, designs for choice experiments, graph and network problems, and cases where optimal designs are used for a sub-selection from a finite database (e.g., \cite{RobCook}, \cite{screening}, \cite{block}, \cite{weighing}, \cite{choice}, \cite{networks}, \cite{iboss}, \cite{subsamp}).

\textbf{Example.} Consider a dose-response experiment in a clinical trial, where the design space is the set of all feasible doses that can be administered to patients (or volunteers).  In this setting, each dose must be assigned to an integer number of patients, making direct implementation of the real-valued proportions prescribed by an approximate design infeasible. Additionally, the realistic design space is finite, consisting of a predetermined set of doses selected based on prior clinical practice and the precision with which the dose can be measured (see, e.g., \cite{Ting}, Ch. 7.3.4.).

This paper therefore focuses on exact designs.\footnote{Exact designs are also called 'integer' designs or 'finite-sample' designs.} Each exact design is formalized as a function $w : \X \to \{0,1,2,\ldots\}$ on a finite design space $\X = \{x_1, \ldots, x_n\}$ , where $w(x)$ directly determines the number of replicated trials at the candidate point $x \in \X$. Note that in most practical experiments $w(x)>0$ for only a small fraction of candidate points, and the set of all such points is called the support of $w$.

Furthermore, in many modern experimental settings, observed responses are multivariate, and acceptable designs must satisfy multiple constraints. These constraints may stem from realistic restrictions, such as finite experimental resources, safety requirements, and the need for balance or fairness, resulting in a complex set $\Xi$ of permissible designs.

\textbf{Example.} Consider a dose-response experiment aimed at finding an optimal exact design $w^*$ for a bivariate dose-response model. The design must be constrained to a total of $75$ patients, with no more than $10$ expected toxic or non-efficacious responses. Additionally, the design should include $5$ different dose levels, each administered to at least $10$ and at most $20$ patients. Importantly, the optimization process must itself select the $5$ different dose levels from the full design space $\X$ to maximize the criterion value. 

In scenarios involving complex models and multiple design constraints, analytic solutions for optimal designs are rarely available, necessitating numerical methods. We therefore develop a computational approach for such problems.

More precisely, the main contribution of this paper is to show that any of the available mathematical programming (MP) procedures capable of computing optimal exact designs for univariate-response models under linear constraints can be extended to compute optimal exact designs for multivariate-response models, which can in addition be restricted by a flexible hybrid class of constraints. These constraints combine both linear and sparsity restrictions on the design, such as those in the example above. 
This extension significantly broadens the applicability of various existing MP-based procedures (e.g., \cite{aqua}, \cite{aqua1}, \cite{milp}, \cite{SagnolHarman}, \cite{rlib}).

The key technical idea is to transform the primary optimal design problem (a multivariate-response optimal design problem under linear and sparsity constraints) to an auxiliary optimal design problem that can be solved using existing MP procedures. This auxiliary problem is straightforward to set up, and its size, in terms of the number of variables and constraints, is only moderately larger than that of the original problem. The optimal design of the auxiliary problem can then be used to directly obtain an optimal design for the primary problem.

This paper is structured as follows. Sections \ref{sec:constraints} and \ref{sec:problem} introduce the specific variants of the optimal design problem addressed in this work. We begin by proposing a comprehensive classification of the various linear constraints that have been studied in the context of optimal experimental design. We then discuss novel linear and sparsity (LAS) constraints on the set of exact designs over a finite design space, as well as statistical models with multivariate responses. In the key Section \ref{sec:computation} we demonstrate how the problem of LAS constraints for multivariate-response models can be converted to the case that can be solved by existing mathematical programming methods. Finally, in Section \ref{sec:examples}, we present examples of how the proposed method can be used to construct statistically efficient and directly implementable experiments for dose-response studies, sequentially incorporating various realistic conditions related to efficacy, safety, and costs. We emphasize that the same methods can be applied to many other optimal experimental design problems with similar types of constraints, such as screening factorial experiments (e.g., \cite{woods}), block designs for comparative trials (\cite{bailey}), optimal designs for accelerated lifetime testing (such as \cite{alt}), and beyond.

\section{Experimental design constraints}\label{sec:constraints}

The experimental design is sometimes restricted to only a subset of a ``full'' design space, for instance, when some (combinations of) factor levels are not allowed. However, such domain constraints pose no major issue if we consider a finite design space $\X$. Nontrivial experimental constraints that we deal with in this section are conditions on the set $\Xi$ itself, i.e., on the set of functions $w:\X \to \{0,1,2,\ldots\}$ representing permissible designs.

The most prevalent restriction in optimal experimental design is the standard constraint on the experimental size, formally $\sum_{i=1}^n w(x_i)=N$. That is, $N$ is a given required number of trials conducted as part of the experiment. The size constraint is often the only constraint imposed on the design. However, other types of restrictions can lend the experimenters substantially broader control over the experiment. An important class of restrictions can be formalized as $K$ general linear constraints of the form
\begin{equation}\label{eq:XiLin}
    \sum_{i=1}^n a(x_i,k) w(x_i) \leq b(k),
\end{equation}
for all $k \in \{1:K\}:=\{1,\ldots,K\}$, where $a:\X \times \{1:K\} \to \RR$ and $b:\{1:K\} \to \RR$ are given functions. We denote the set of all designs $w$ satisfying \eqref{eq:XiLin} by $\Xi(\X,a,b)$.

In this paper we do not require that the constraints \eqref{eq:XiLin} involve the size constraint, i.e., the optimal size of the design may also be subject to the optimization process. Nevertheless, in applications the constraints \eqref{eq:XiLin} always make the set $\Xi(\X,a,b)$ bounded, because each trial consumes some kind of resources that are ultimately limited.

 For clarity, we suggest classifying the various linear constraints of the form \eqref{eq:XiLin} that have appeared in the literature into three families, based on their role. A constraint with index $k \in \{1:K\}$ can be: 
\begin{enumerate}
    \item \textbf{An 'exclusion', or 'resource', constraint}, which generally serves to limit trials. Such a constraint is characterized by inequalities $a(x,k) \geq 0$ for all $x \in \X$ and $b(k)>0$ (see \cite{HBF}). A special case is a 'privacy' constraint (\cite{Benkova2016}), defined on a privacy region $\tilde{\X} \subseteq \X$, where $a(x,k) = 1$ if $x \in \tilde{\X}$, and $a(x,k) = 0$ if $x \notin \tilde{\X}$. If $\tilde{\X}$ is a singleton, i.e., $\tilde{\X}=\{x\}$, we obtain a ``direct'' constraint, limiting the number of replications at a specific point $x$ (e.g., \cite{Fedorov82}). A system of privacy constraints over a partition of $\X$ defines strata constraints (\cite{Harman2014JSPI}). In factor experiments, privacy regions may correspond to sets of points sharing a level of a factor, resulting in marginal constraints (e.g., \cite{CT}). Different  structure-enforcing privacy constraints also include Latin-hypercube constraints (\cite{McKay}) and, more generally, bridge constraints (\cite{Jones2014}). In Section \ref{sec:examples}, we show that the exclusion constraints, which embody toxicity constraints, can limit toxic outcomes in the context of dose-response studies.

    \item \textbf{An 'inclusion' constraint}, which enforces the presence of trials at certain points or regions in the design space. In our formalization, an inclusion-type constraint corresponds to a linear inequality with $a(x,k) \leq 0$ for all $x \in \X$ and $b(k)<0$. The most important class of inclusion-type constraints are those that require minimum number of replications at design points, i.e., the constraint of the type $w(x_i) \geq N(x_i)$, where $N(x_i)$ is a given lower replication limit, $i \in \{1:n\}$. These constraints appear in the literature under various names, for instance 'protected runs' (Section 11.7 in \cite{Puk}) or 'augmentation' of trials (cf. Chapter 19 in \cite{atkinson}). However, inclusion constraints can also enforce lower bounds on specific design features, i.e., minimum replication numbers within entire subregions $\tilde{\X} \subseteq \X$ (e.g., Example 6.3 in \cite{aqua}). In Section \ref{sec:examples} we show that the inclusion constraints can represent minimum efficacy requirements in dose-response studies.

    \item \textbf{A 'mixed' constraint}, which combines positive and negative weights in $a(x,k)$, that is, $a(x_1,k) > 0$ and $a(x_2,k) < 0$ for some $x_1, x_2 \in \X$. These constraints typically encode balance or comparative weighting between subgroups of design points. For example, suppose $\X$ includes covariates indicating patient gender, with $\tilde{\X}_F$ and $\tilde{\X}_M$ denoting the subsets corresponding to females and males, respectively. To enforce equal allocation across genders, one can impose the constraint $\sum_{i: x_i \in \tilde{\X}_F} w(x_i) - \sum_{i: x_i \in \tilde{\X}_M} w(x_i) = 0$.
\end{enumerate}

Importantly, there exist constraints of potentially high practical utility other than \eqref{eq:XiLin}, such as constraints on the number and distribution of points at which measurements will be taken, without restricting the design space itself, which we refer to as sparsity constraints. Formally, these are constraints involving properties of the support set of the design $w$. One of the main new contributions of this paper is to show how existing computational methods can be used in situations involving both linear and sparsity constraints.

Formally, by linear and sparsity (LAS) constraints for designs of size $N$ we understand
\begin{eqnarray}\label{eq:XiLas}
    \sum_{i=1}^n a(x_i,k)w(x_i) + \sum_{i=1}^n c(x_i,k)s_w(x_i) &\leq& b(k) \: \text{ for } \: k \in \{1:K\},\\
    \sum_{i=1}^n w(x_i) &=& N, \nonumber
\end{eqnarray}
where $a,c : \X \times \{1:K\} \to \RR$, and $b:\{1:K\} \to \RR$ are given coefficient functions and $s_w=\mathrm{sgn}(w)$, i.e., $s_w$ is the indicator of the support of $w$: $s_w \in \{0,1\}$ and $s_w(x)=1 \Leftrightarrow w(x)>0$. We will denote the set of all designs $w$ satisfying constraints \eqref{eq:XiLas} by the symbol $\Xi_N(\X,a,c,b)$.

The LAS constraints \eqref{eq:XiLas} allow the experimenter to require special additional properties of the design that cannot be directly written as linear constraints, for instance:

\begin{itemize}
        \item \textbf{Maximum support size.} Let $K=1$, $a(x,1) \equiv 0$, $c(x,1) \equiv 1$ and $b(1) = S \in \NN$, $S \leq N$. Then $\Xi_N(\X,a,c,b)$ is the set of all designs of size $N$ with the support size not exceeding $S$. That is, the permissible experiments use exactly $N$ trials but at most $S$ different trial conditions, for example at most $S$ different doses in a dose-response experiment.
        \item \textbf{Support-influenced cost.} Let $K=1$, $a(x,1) \geq 0$ represent costs for each trial at $x \in \X$, next $c(x,1) \geq 0$ represent overhead costs for each new support point and $b(1)=B>0$ is a budget limit. Then $\Xi_N(\X,a,c,b)$ is the set of all designs of size $N$ fitting to the experimental budget $B$.
        \item \textbf{Support points separation.} Let $\X_1,\ldots,\X_K \subseteq \X$, $a(x,k) \equiv 0$, $c(x,k) \equiv 1$ for all $x \in \X_k$, $c(x,k) \equiv 0$ for all $x \notin \X_k$ and $b(k) \equiv 1$. Then the support set of any design in $\Xi_N(\X,a,c,b)$ has at most one element in any of the subsets $\X_1,\ldots,\X_K$. By an appropriate choice of these sets one can achieve desired ``space-filling'' separation of the support points. 
        \item \textbf{Limits of replications at support points.} Let $K=2n$, and $a,c,b$ are chosen to represent $w(x_i) \geq L(x_i)s_w(x_i)$ and $w(x_i) \leq U(x_i)s_w(x_i)$ for all $i \in \{1:n\}$, where $L(x)$ and $U(x)$ are given potential replication limits for each $x \in \X$. Then a design $w$ belongs to $\Xi_N(\X,a,c,b)$ if and only if $w$ is a size-$N$ design that replicates observations in each design point $x \in \X$ either $0$ times, or between $L(x)$ and $U(x)$ times.
\end{itemize}

\section{The optimal design problem}\label{sec:problem}

Once the set $\Xi$ of permissible experimental designs is defined, we introduce a real-valued optimality criterion $\phi$ on $\Xi$, the maximization of which allows us to identify the best design. We focus on model-oriented optimal design, where the goal is to estimate parameters of a statistical model or a function thereof.

In this framework, the criterion $\phi$ is typically defined as $\phi(w)=\Phi(\M(w))$, where $\Phi$ is a scalar function and $\M(w)$ is the so-called information matrix of the design $w$. The matrix $\M(w)$ equals or approximates (a fixed multiple of) the Fisher information matrix for the unknown parameter vector $\theta$ in the statistical model. In the case of independent responses across trials, the information matrix typically takes the form $\M(w) = \sum_{i=1}^n w(x_i)\H(x_i)$, where $\H(x_1),\ldots,\H(x_n)$ are known positive semidefinite 'elementary' information matrices.\footnote{The elementary information matrices are sometimes also called 'elemental' information matrices or 'one-point' information matrices.} The elementary information matrices are based on the statistical model used.

In the univariate-response models, the elementary information matrices have rank one, i.e., $\H(x_i)=\f(x_i)\f^\top(x_i)$, where $\f(x_1),\ldots,\f(x_n) \in \mathbb{R}^m$ are known $m$-dimensional vectors. For instance, consider the univariate-response nonlinear regression model $Y(x) = \eta(x,\theta) + \epsilon$, $\epsilon \sim \mathcal{N}(0, \sigma^2)$, where $\eta:\X \times \mathbb{R}^m \to \mathbb{R}$ is a smooth function and $\theta \in \mathbb{R}^m$ and $\sigma>0$ are parameters. Then the elementary information matrices can be defined using the local optimality principle and asymptotic approximations at $\theta_0$ (\cite{chernoff}) as $\H(x)=\f(x)\f^\top(x)$, where $\f(x)$ is the gradient of $\eta(x, \theta):\mathbb{R}^m \to \mathbb{R}$ with respect to $\theta$ evaluated at $\theta_0$.

The optimality criterion $\Phi$ measures the size of positive semidefinite matrices. The most popular criterion is $D$-optimality, which in its information-function form (\cite{Puk}) can be defined as: $\Phi_D(\M) = (\det(\M))^{1/m}$ for any positive semidefinite $\M$. However, our approach also immediately applies to a wide class of other criteria, such as $D_K$-, $A_K$-, $c$-, $I$-, $G$-, $MV$-optimality, and convex combinations of these criteria, for which mathematical programming-based tools are available.

The optimal exact design problem that is most commonly considered is that of univariate-response models with linear constraints on the designs, formally
\begin{equation}\label{eq:Lproblem}
w^* \in \mathrm{argmax}\left\{ \Phi\left(\sum_{i=1}^n w(x_i)\f(x_i)\f^\top(x_i)\right) \: : \: w \in \Xi(\X,a,b) \right\}.  
\end{equation}

Computational methods for handling various linear constraints have been extensively studied in the context of approximate design theory. In that setting, convex optimization techniques are available, see, e.g., \cite{CookFedorov}, \cite{CookWong}, \cite{Harman2014JSPI}, \cite{HarmanBenkova2017}, \cite{sdp}, \cite{sagnol} and others. In contrast, relatively few methods target the \emph{exact} optimal design problem \eqref{eq:Lproblem} under nontrivial linear constraints. These include heuristic methods tailored to specific constraint families (e.g., resource constraints in \cite{HBF} or privacy constraints in \cite{Benkova2016}), and, most often, mathematical-programming based methodologies (\cite{DGW}, \cite{duarte}, \cite{SagnolHarman}, \cite{aqua1}, \cite{aqua}, \cite{milp}).

While most optimal design literature focuses on optimal experimental design for univariate-response models, many real-world experiments involve multivariate outcomes of each individual trial. Multiresponse models arise in various applications, including chemical engineering, clinical trials, and drug development. One of the first references to optimal design for such models is \cite{DraperHunter}, where augmented $D$-optimality is proposed. Next, \cite{Fed71} provided theoretical foundations, including an equivalence theorem for linear multiresponse models. Approximate optimal designs in multivariate models have been studied, e.g., in \cite{dragalin}, where a gradient descent method is applied to the bivariate binary Cox model, and in \cite{dfw}, where a bivariate probit model is used in a similar setting.

The primary problem addressed in this paper is therefore to find the exact optimal design for the multivariate-response setting under the versatile linear and sparsity (LAS) constraints introduced earlier. Formally, we aim to solve
\begin{equation}\label{eq:LASproblem}
w^* \in \mathrm{argmax}\left\{ \Phi\left(\sum_{i=1}^n w(x_i)\H(x_i)\right) \: : \: w \in \Xi_N(\X,a,c,b) \right\},    
\end{equation}
where $\H(x_1),\ldots,\H(x_n)$ are general-rank positive semidefinite $m \times m$ matrices and $\Xi_N(\X,a,c,b)$ is the set of all LAS-constrained designs of size $N$. In the next section, we show that \eqref{eq:LASproblem} can be solved by leveraging standard mathematical programming tools developed for solving \eqref{eq:Lproblem}.

\section{Computing multiresponse LAS-constrained optimal designs}\label{sec:computation}

The key strategy to solving the primary problem \eqref{eq:LASproblem} is to transform it to a problem of the form \eqref{eq:Lproblem} defined by auxiliary regressors $\f'$, design space $\X'$ and linear constraint functions $a',b'$.

\begin{thm}\label{thm:main}
Consider the primary optimal design problem \eqref{eq:LASproblem} for a multivariate-response model on a design space $\X$ of size $n$ defined by the elementary information matrices $\H(x)$, $x\in \X$, with the linear and sparsity constraints $\Xi_N(\X,a,c,b)$. 
Let $r \in \NN$ be such that $\rank(\H(x_i)) \leq r$ for all $i \in \{1:n\}$. 
Let $w'^*$ be a solution of the following auxiliary univariate-response model of the type \eqref{eq:Lproblem} on a design space $\X'$ of size $n'=nr+n$ defined by regressors $\f'$ and linear constraints $\Xi(\X',a',b')$:
\begin{itemize}
    \item $\X'=\big(\X \times \{1:r\}\big) \cup \mathcal{Z}$, where $\mathcal{Z}=\{z_1,\ldots,z_n\}$ is any $n$-point set of independent labels. 
    
    \item For all $i\in \{1:n\}$ the regressors $\f'(x_i,1),\ldots,\f'(x_i,r)$ satisfy $\H(x_i)=\sum_{j=1}^r \f'(x_i,j)\f'^\top(x_i,j)$.

    \item For all $i\in \{1:n\}$ the regressors $\f'(z_i)$ are the $m$-dimensional zero vectors.
 
    \item The values $a'((x_i,j),k)$, $a'(z_i,k)$, $b'(k)$ for $i \in \{1:n\}$, $j \in \{1:r\}$ and $k \in \{1:K'\}$ are chosen such that they correspond to the following restrictions on the auxiliary exact designs $w'$ on $\X'$:
    \begin{eqnarray}
        && w'(x_i,1)=\cdots=w'(x_i,r) \text{ for all } i \in \{1:n\},\label{eq:1}\\
        && w'(z_i) \leq 1 \text{ for all } i\in \{1:n\},\label{eq:2}\\
        && w'(z_i) \leq w'(x_i,1) \leq N w'(z_i) \text{ for all } i\in \{1:n\},\label{eq:3}\\
        && \textstyle \sum_{i=1}^n a(x_i,k) w'(x_i,1) + \sum_{i=1}^n c(x_i,k)w'(z_i) \leq b(k) \text{ for all } k\in \{1:K\},\\
        && \textstyle \sum_{i=1}^n w'(x_i,1) = N.
    \end{eqnarray}
\end{itemize}
Then $w^*(\cdot) = w'^*(\cdot,1)$ is the optimal solution to the primary optimal design problem \eqref{eq:LASproblem}.
\end{thm}

\begin{proof}
    For each exact design $w'$ on $\X'$ let $\kappa(w')$ be an exact design on $\X$ defined by  $\kappa(w')(x_i)=w'(x_i,1)$, $i \in \{1:n\}$. To prove the theorem, it is enough to show that
    \begin{itemize}
        \item[(a)] $\kappa$ is a bijective mapping from $\Xi(\X',a',b')$ to $\Xi_N(\X,a,c,b)$, i.e., a design $w' \in \Xi(\X',a',b')$ is in one-to-one correspondence with the design $\kappa(w') \in \Xi_N(\X,a,c,b)$, and
        \item[(b)] the objective function of the auxiliary problem  evaluated in any $w' \in \Xi(\X',a',b')$ coincides with the objective function of the primary problem evaluated in  $\kappa(w') \in \Xi_N(\X,a,c,b)$, that is,
        \begin{equation*}
        \Phi\left(\sum_{i'=1}^{n'} w'(x'_{i'})\f'(x'_{i'})\f'^\top(x'_{i'})\right) = \Phi\left(\sum_{i=1}^n \kappa(w')(x_i)\H(x_i)\right).
        \end{equation*}
    \end{itemize} 
    To show that $\kappa$ is injective, assume that $w'_1, w'_2 \in \Xi(\X',a',b')$ satisfy $\kappa(w'_1)=\kappa(w'_2)$, that is, $w'_1(x_i, 1)=w'_2(x_i, 1)$ for all $i \in \{1:n\}$. Then, clearly, condition \eqref{eq:1} implies that $w'_1$ and $w'_2$ coincide on $\X \times \{1:r\}$. However, conditions \eqref{eq:2} and \eqref{eq:3} imply that $w'_1$ and $w'_2$ also coincide on $\mathcal{Z}$. Indeed, for any $i \in \{1:n\}$ we have 
    \begin{equation*}
        w'_1(z_i)=\mathrm{sgn}(w'_1(z_i))=\mathrm{sgn}(w'_1(x_i,1))=\mathrm{sgn}(w'_2(x_i,1))=\mathrm{sgn}(w'_2(z_i))=w'_2(z_i),
    \end{equation*}
    where the first and the fifth inequalities follow from \eqref{eq:2}, the second and the fourth equalities follow from \eqref{eq:3} and the third inequality from the assumption $\kappa(w'_1)=\kappa(w'_2)$. Hence, $\kappa$ is injective. To prove that $\kappa$ is surjective, take any $w \in \Xi_N(\X,a,c,b)$; we need to show that $w=\kappa(w')$ for some $w' \in \Xi(\X',a',b')$. But, as can be easily verified, such a $w'$ can be defined as follows: $w'(x_i,1)=\cdots=w'(x_i,r):=w(x_i)$ and $w'(z_i):=\mathrm{sgn}(w(x_i))$ for all $i \in \{1:n\}$. We proved (a).

    Statement (b) follows from
    \begin{eqnarray*}
        \sum_{i'=1}^{n'} w'(x'_{i'})\f'(x'_{i'})\f'^\top(x'_{i'}) &=& \sum_{i=1}^{n} \sum_{j=1}^r w'(x_i,j)\f'(x_i,j)\f'^\top(x_i,j) + \sum_{i=1}^{n}  w'(z_i)\f'(z_i)\f'^\top(z_i) \\
        &=&  \sum_{i=1}^{n} w'(x_i,1) \left(\sum_{j=1}^r \f'(x_i,j)\f'^\top(x_i,j) \right) + \mathbf{0}_{m \times m} \\
        &=& \sum_{i=1}^n \kappa(w')(x_i)\H(x_i).
    \end{eqnarray*}
\end{proof}

The auxiliary regressors $\f'(x_i,j)$ in Theorem \ref{thm:main} can be constructed, for instance, as follows. For each $i \in \{1:n\}$ and the elementary information matrix $\H(x_i)$ let $\lambda_{i,1} \geq \ldots \geq \lambda_{i,m}$ be the eigenvalues, repeated according to their multiplicities, and let $\u_{i,1}, \ldots, \u_{i,m}$ be the corresponding orthogonal system of eigenvectors. As $\rank(\H(x_i)) \leq r$ we have $\lambda_{i,r+1}=\cdots =\lambda_{i,m}=0$, that is
\begin{equation*}
\H(x_i)=\sum_{j=1}^m \lambda_{i,j} \u_{i,j} \u_{i,j}^\top=
	\sum_{j=1}^r(\sqrt{\lambda_{i,j}}\u_{i,j})(\sqrt{\lambda_{i,j}}\u_{i,j})^\top.    
\end{equation*}
We can thus set $\f'(x_i,j)=\sqrt{\lambda_{i,j}}\u_{i,j}$. An alternative decomposition $\H(x_i)=\sum_{j=1}^r \f'(x_i,j)\f'^\top(x_i,j)$ can be obtained from the pivoted Cholesky decomposition, which is also fast and numerically stable (cf. \cite{Higham1990}).

\section{Optimal designs for the continuation ratio model in clinical trials}\label{sec:examples}

In clinical trials, there are generally two conflicting objectives. One is to avoid doses that are either ineffective or cause serious adverse effects (individual ethics); designs that prioritize this objective are often called 'best intention' designs. The other objective is to gain as much information from the experiment as possible (collective ethics), which suggests using statistically efficient, or even optimal, designs.

In Phase I/II trials, safety and efficacy outcomes are investigated jointly within the same study (\cite{sverdlov}, \cite{dragalin}). This accelerates drug development by bridging the gap between the first and second phases, improves the procedure for identifying the target dose by maximizing statistical information while controlling for toxicity, and increases the sample size by combining the phases. This leads to improved precision in estimating both efficacy and toxicity, compared to conducting the phases separately.

However, besides safety and efficacy, experimenters must also consider several other objectives when designing an experiment. For example, it is necessary to ensure that the therapeutic window (the interval of doses considered both safe and efficacious) is sufficiently covered. This creates a requirement for a minimum number of support points in the design while, at the same time, keeping the doses sufficiently dispersed. On the other hand, manufacturing too many different doses may be costly, which favors keeping the number of support points reasonably small.

In this section, we will show how to incorporate these kinds of practical requirements into computing exact designs for continuation ratio model, that is a trinomial response model often used to simultaneously model efficacy and toxicity. The same approach can be used not only in other models for modeling efficacy and toxicity, such as Cox bivariate binary model (\cite{dfw}, \cite{dragalin}), or more generally, multivariate logistic model (\cite{Bu}), but in a variety of other applications that call for similar requirements. 

Consider a finite design space $\mathcal{X}=\{x_1,\ldots,x_n\}$ corresponding to the doses of a drug of interest. We can model the binary efficacy\footnote{Note that by \emph{efficacy} we mean the positive effect of the drug on the patient, in contrast to \emph{efficiency} that is used to describe the quality of the designs.} response at dose $x \in \X$ as the Bernoulli random variable
\[
Y_E(x)=\begin{cases} 1 &\mbox{if the dose } x \mbox{ is efficient} \\
	0 & \mbox{otherwise, } \end{cases} 
\] 

and binary toxicity response at dose $x \in \X$ as the Bernoulli random variable
\[
Y_T(x)=\begin{cases} 1 &\mbox{if the dose } x \mbox{ is toxic} \\
	0 & \mbox{otherwise. } \end{cases} 
\]

At dose $x$, the probability that the patient has no reaction is $p_0(x)=P[Y_E(x)=0, Y_T(x)=0]$. Similarly, the probability of efficacy without toxicity will be denoted by $p_S(x)=P[Y_E(x)=1, Y_T(x)=0]$ and the probability of toxicity by $p_T(x)=P[Y_T(x)=1]$. In a parametric efficacy-toxicity model, the distributions of $Y_E(x)$ and $Y_T(x)$ also depend on a vector $\theta$ of parameters, which results in a dependence of the probabilities $p_0$, $p_S$ and $p_T$ on $\theta$, that is, $p_0(x)=p_0(x,\theta)$, $p_S(x)=p_S(x,\theta)$,  and $p_T(x)=p_T(x,\theta)$. In this paper, we model these probabilities by the continuation-ratio (CR) model given by
\begin{equation}\label{CR}
\begin{split}
\log\frac{p_{T}(x,\theta)}{1-p_{T}(x,\theta)}&=a_1+b_1x,\\
\log\frac{p_{S}(x,\theta)}{p_{0}(x,\theta)}&=a_2+b_2x,
\end{split}
\end{equation}
where the parameter vector is $\theta=(a_1,a_2,b_1,b_2)^T$ with $a_1\in\mathbb{R},a_2\in\mathbb{R},b_1>0,b_2>0$. Taking the necessary requirement $p_{T}(x,\theta)+p_S(x,\theta)+p_0(x,\theta)=1$ into account, model \eqref{CR} implies
\begin{eqnarray*}
p_{0}(x,\theta)&=&(1+e^{a_1+b_1x})^{-1}(1+e^{a_2+b_2x})^{-1},\\
p_{S}(x,\theta)&=&e^{a_2+b_2x}(1+e^{a_1+b_1x})^{-1}(1+e^{a_2+b_2x})^{-1},\\
p_{T}(x,\theta)&=&e^{a_1+b_1x}(1+e^{a_1+b_1x})^{-1}.
\end{eqnarray*}

Approximate locally $D$- and $c$-optimal designs for this model were found by \cite{FanChaloner}. Further, \cite{CRmodel} used particle swarm optimization to compute compound approximate optimal designs for estimating the parameters of the model and the most effective dose. In \cite{Rabie}, continuation ratio model is generalized to contingent ratio model and approximate optimal designs are computed. In \cite{Bu}, the authors consider a general situation of the multinomial logistic model that also covers the continuation ratio model and compute approximate, exact\footnote{The exact designs are computed using the exchange algorithm, and only the size constraint is considered.} and Bayesian $D$-optimal designs. 
\qquad

Denote $\f_1(x) = (1,x,0,0)^\top$ and $\f_2(x) = (0,0,1,x)^\top$. According to \cite{FanChaloner} (cf. \cite{Bu}), the elementary information matrix at the point $x$ can be expressed in the form
\begin{equation}\label{CRM}
    \H(x,\theta) = \frac{e^{a_2+b_2x}}{(1+e^{a_2+b_2x})^2(1+e^{a_1+b_1x})} \f_1(x)\f_1^\top(x) + \frac{e^{a_1+b_1x}}{(1+e^{a_1+b_1x})^2}\f_2(x)\f_2^\top(x).
\end{equation}

\qquad

Following \eqref{CRM}, we can see that the information matrix in the dose $x$ for the continuation ratio model is a rank-two matrix for every $x$ and $\theta$, and can be written as a weighted sum of two rank-one matrices. For the computation of optimal designs, the standard approach in these models is the approach of locally optimal designs (see \cite{chernoff} and \cite{pronzatopazman2013}). This means that, in our case, $\H(x) = \H(x, \theta_0)$, where $\theta_0$ denotes the nominal value of the parameter $\theta$. For computing locally optimal designs in the continuation ratio model with linear and sparsity constraints, we use Theorem~\ref{thm:main}.

Although \cite{atk-bis} deals with design randomization in later phases of clinical trials, some of the constraints considered in Chapter 8, such as minimizing the expected total failure or minimizing the expected value of a loss function, are relevant in our case as well. 

\subsection{Constrained designs for efficacy and toxicity}

Consider the continuation ratio model \eqref{CR} with nominal values of the parameters $a_1=-9.5$, $a_2=-9.1$, $b_1=0.12$, $b_2=0.33$ and the dose range from 0 to 100 discretized into 101 equidistant doses, that is, $\mathcal{X}=\{0,1,2,\ldots, 99,100\}$. 
We want to compute exact $D$-optimal designs when $N=100$ patients are available. Given the approximate nature of the optimal design methodology, one should preferably suggest more designs to choose from. In this study, we considered several practical LAS-constraints (see below). Based on Theorem \ref{thm:main}, we transformed the primary multivariate-response problem into the corresponding auxiliary univariate-response problem, which we solved using the mixed-integer second-order cone programming (MISOCP) method (\cite{SagnolHarman}) implemented in the \texttt{R} package \texttt{OptimalDesign} (see \cite{rlib}).  All computations were performed on a computer with a $64$-bit Windows 11 operating system running an AMD Ryzen 7 5800H CPU processor at 3.20 GHz with 16 GB RAM. The time required for the computations ranged from $5$ to $40$ seconds across the different scenarios.

In the following subsections, we describe 6 different cases for which we compute locally optimal designs (these will be denoted as $w_0^*,w_1^*,\dots, w_5^*$). The constraints in each case are added sequentially, meaning that a new aspect is added to the previous one in each step. Note that for comparing the statistical performance of two designs $w_1$ and $w_2$, we use the $D$-efficiency $\mathrm{eff}(w_1 | w_2) = [\det(\M(w_1))/\det(\M(w_2))]^{1/m}$, which is
interpretable in terms of the relative number of trials required to reach the same criterion value. 

It should be noted that these experiments have an illustrative character and are not based on practical implementation; they are scenarios that demonstrate the theoretical results and possibilities opened up by this paper.

If the only constraint is the limit on the number of patients $N=100$, the exact locally $D$-optimal design $w_0^*$ puts patients into 6 different doses, as can be seen in Table \ref{Table:designs} and Figure \ref{Fig:designs_all_gg}. 

\subsubsection{Constraint on the expected number of failed trials}

For the design to meet safety requirements, it can be beneficial to maximize the probability of success $p_{S}$, i.e., minimize the number of failures, that is, doses that have either adverse effects and/or are ineffective. 
Following \eqref{CR}, the failure probability is
\[
p_f(x,\theta)=1-p_{S}(x,\theta)=p_0(x,\theta)+p_T(x,\theta).
\] 
In the design $w_0^*$, the expected number of failed trials is $49.35$. Assume that this value is deemed unacceptably high, and the maximum number of failed trials should be enforced to be maximum $40$. 

Now, if we consider constraints that the drug can be given to at most $N$ subjects and the expected number of failed trials is at most $u$, these constraints can be expressed in the form \eqref{eq:XiLas} by putting $a(x_i,1)=p_f(x_i,\theta)$, $b(1)=u$. Note that in the approach of local optimality, the coefficients of such constraints are also localized at a nominal parameter value $\theta_0$.

In our case, setting $N=100$ and $b(1)=40$ yields the design $w_1^*$ shown in Table \ref{Table:designs} and Figure \ref{Fig:designs_all_gg}. The relative efficiency of $w_1^*$ with respect to the basic design $w_0^*$ remains very high -- 98\%.

\subsubsection{Constraints on the total cost of the experiment with support-influenced costs}\label{subsubsec:CR_costs}

Assume that a cost $\gamma(x_i) \geq 0$ is associated with administering the dose $x_i$, $i \in \{1:n\}$; and a budget $B$ is allocated to the experiment.
This constraint can be expressed as a linear constraint $a(x_i,1) = \gamma(x_i)$ and $b(1)=B$ in Theorem \ref{thm:main}.

Suppose that, in addition to the costs associated with administering dose $x_i$ to $w(x_i)$ patients, it is also possible to face overhead costs for preparing each dose. However, these overhead costs are not charged for each patient who receives the dose but only once, to ensure the preparation of the dose $x_i$. Denote these overhead costs for preparing the dose $x_i$ as $\gamma'(x_i)$. Then, the LAS constraint can be expressed as $c(x_i,1) = \gamma'(x_i)$ in addition to the original cost constraints $a(x_i,1) = \gamma(x_i)$ with $b(1)=B$.

Assume that the cost of manufacturing and distributing a dose $x$ be $0.4x$. In addition, suppose that we must account for the cost of incorrectly dosing patients: an under-dosed patient (receiving an ineffective dose) incurs a cost of $5$, which may be associated with the need to treat the patient later with a different drug. An over-dosed patient (experiencing adverse effects) incurs a cost of $20$, reflecting the treatment required for those adverse effects. All other patients incur zero additional cost. Under these assumptions, the designs $w_0^*$ and $w_1^*$ would cost $712$ and $598$ units, respectively. If we now add a total budget constraint of $B = 500$ to the size and failure constraints, we obtain the design $w_2^*$ in Table \ref{Table:designs} and Figure \ref{Fig:designs_all_gg}, which achieves $96\%$ efficiency relative to the original design $w_0^*$.

\subsubsection{Constraint on the minimum or maximum support size}

The design $w_2^*$, while reasonably safe and cost-effective, may be considered insufficient by some practitioners, as it includes only four different doses. More generally, there may be a requirement to ensure that no fewer than $S$ different doses are tested during the clinical trial. This can be achieved by setting $a(x_i,1) = 0$, $c(x_i,1) = -1$ for $i \in \{1:n\}$, and $b(1) = -S$. Similarly, it is also possible to specify an upper bound on the number of distinct doses tested, requiring that no more than $S$ ($S<n$) different doses be included in the design. 

Imposing a lower bound of $S = 6$ on the number of different doses, in addition to the constraints previously considered, yields the design $w_3^*$ in Table \ref{Table:designs} and Figure \ref{Fig:designs_all_gg}, which has a relative efficiency of $95$\% with respect to $w_0^*$.

\subsubsection{Enforcing the space-fillingness of the design}\label{subsubsec:uniformization}

While the design $w_3^*$ technically satisfies the requirement of having support on at least $6$ points, it may fall short in terms of practical usefulness, as it includes three doses that are very close to each other. To address this problem, it would be reasonable to introduce a constraint that enforces a minimum distance $\Delta$ between any two design points included in the support, thus avoiding such clustering.

Formally, given the dose levels $x_1, x_2, \dots, x_n$, the following LAS constraints can be formulated:
\begin{eqnarray}
    c(x_{1}, 1) &=& c(x_{2}, 1) = \dots = c(x_{\Delta}, 1) = 1  \label{eq:unif1} \\
    c(x_{2}, 2) &=& c(x_{3}, 2) = \dots = c(x_{(\Delta+1)}, 2) = 1  \label{eq:unif12} \\ 
    &\vdots&  \nonumber \\
    c(x_{(n-\Delta+1)}, n-\Delta+1) &=& \dots = c(x_{n}, n - \Delta+1) = 1  \label{eq:unif3}
\end{eqnarray}

Applying this constraint with the minimum distance $\Delta = 10$ between any two design points yields the design $w_4^*$ (see Table \ref{Table:designs}), which has a relative efficiency of $94\%$ with respect to $w_0^*$.

\begin{figure}
\centering
\includegraphics[width=0.9\linewidth]{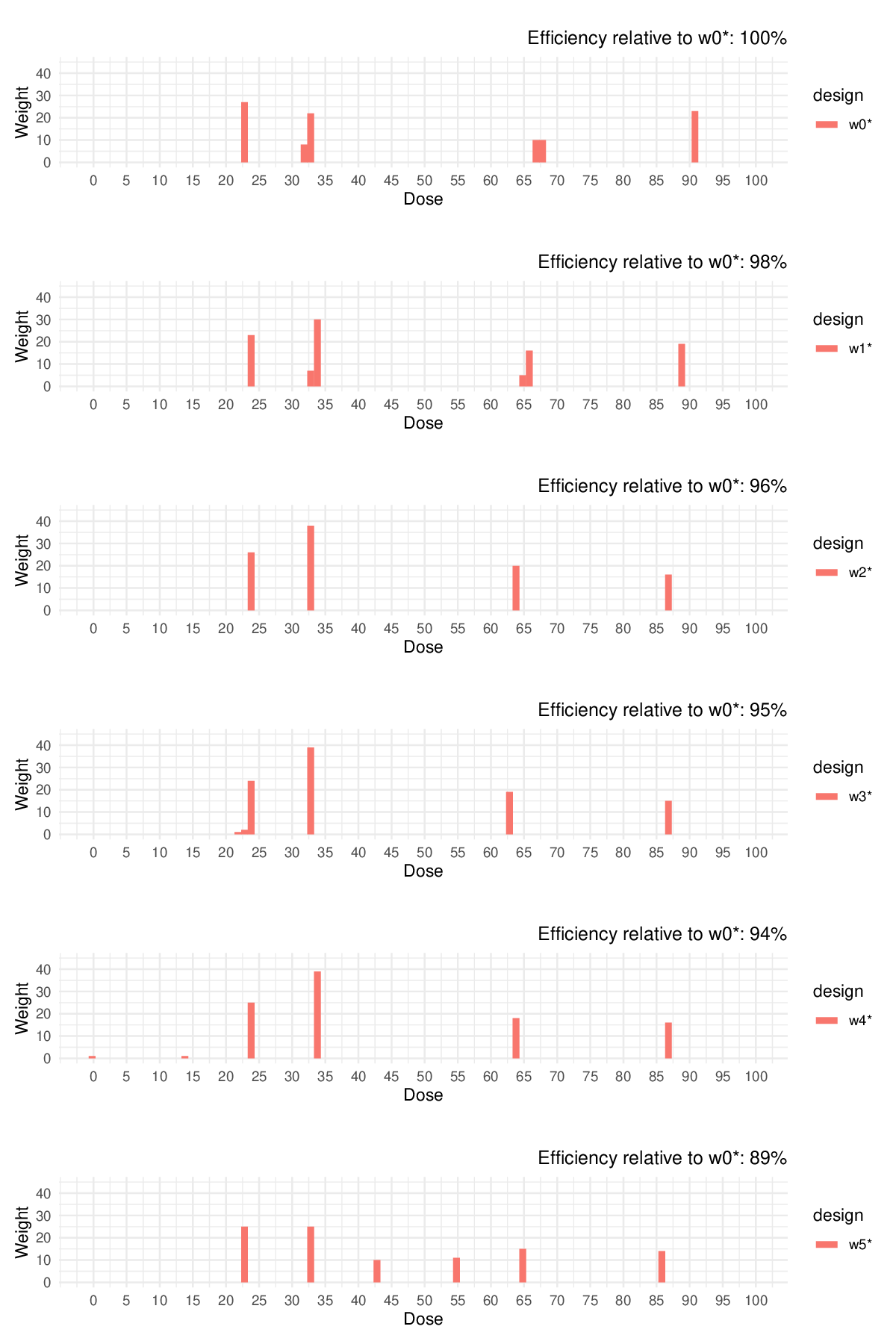}\caption{$D$-optimal designs for model \eqref{CR}. The weight is the integer replication number of each design point, i.e., the number of patients assigned to each dose. The sequence of designs corresponds to sequential adding of practical constraints. Note that the final design $w_5^*$ is $30\%$ less expensive, leads to $12\%$ less failed trials and has a space-filling character with the required numbers of replications at support points; yet its statistical efficiency is only marginally lower than that of the original unconstrained design $w_0^*$. For details, see also Table \ref{Table:designs}.} \label{Fig:designs_all_gg}
\end{figure}

\subsubsection{Constraints on the limits of replications at support points}\label{subsubsec:limits_of_rep}

Another common type of constraint encountered in practice involves bounding the replication numbers by $w(x_i) \in [L(x_i), U(x_i)]$ for all $i \in \{1:n\}$. This can be interpreted as requiring that each dose $x_i$ be administered to at least $L(x_i)$ and at most $U(x_i)$ patients. These conditions correspond to the linear constraints with $a(x_k,k) = -1$, $b(k) = L(x_k)$ for $k \in \{1:n\}$, and $a(x_{k-n},k) = 1$, $b(k-n) = U(x_{k-n})$ for $k \in \{(n+1):2n\}$.

Note that for every $L(x_i) > 0$, the dose $x_i$ must be prepared, regardless of whether it contributes meaningfully to the design. A more practical alternative is to require that, \emph{if} a dose $x_i$ is prepared (i.e., included in the support of the design), it must be administered to at least $L(x_i)$ patients, while still maintaining an upper bound of $U(x_i)$. More formally, this means $w(x_i) = 0$ or $w(x_i) \in \{L(x_i):U(x_i)\}$. These bounded replication conditions at the support points can be encoded using the sparsity constraints
\[
w(x_i) \geq L(x_i) \, s_w(x_i) \quad \text{and} \quad w(x_i) \leq U(x_i) \, s_w(x_i), \quad \text{for } i \in \{1:n\},
\]
which correspond to the following LAS formulation: $a(x_k,k)=-1$, $c(x_k,k)=L(x_k)$ for $k \in \{1:n\}$ and $a(x_{k-n},k)=1$, $c(x_{k-n},k)=-U(x_{k-n})$ for $k \in \{(n+1):2n\}$ from Theorem \ref{thm:main}, with all other $a(x_i,k)$ and $c(x_i,k)$ equal to zero.

Upon examining the design $w_4^*$, we see that this type of constraint is reasonable to apply. Specifically, avoiding the manufacture of doses 0 and 14 may be preferable if they are to be administered to only one patient each. Simultaneously, the relatively high proportion of patients assigned to dose $34$ may be a concern. 

To address these issues, we impose a lower bound $L = 10$ and an upper bound $U = 25$ on the number of patients receiving each dose used in the design. This results in the design $w_5^*$ shown in Table~\ref{Table:designs}.

Note that this design satisfies several practical requirements: in addition to having a total size of $N = 100$, it respects the budget constraint $B = 500$, limits the expected number of failures to at most $30$, ensures a reasonable spread across doses, and avoids extreme imbalance in dose allocation. Despite these practical constraints, the design remains statistically robust, achieving $89\%$ relative efficiency with respect to the original locally $D$-optimal design $w_0^*$ considering only the size constraint $N=100$.

\begin{table}[ht!]
\centering
\begin{tabular}{|c|ccccccc||c|c|c|c|}
\hline
&& design &&&&&& $\Phi$ & eff & E(fail) & cost\\
\hline
$w_0^*$ & x & 23 & 32  & 33  & 67  & 68 &  91 & 60.11 & 1 & 49.35 & 711.80\\
 & w & 27  &  8  & 22 &  10 &  10 &  23 & & & & \\
 \hline
 $w_1^*$ & x & 24 & 33 & 34 & 65 & 66 & 89 & 58.75 & 0.98 & 39.99 & 597.83\\
 & w & 23 & 7 & 30 & 5 & 16 & 19 & & & & \\
 \hline
 $w_2^*$ & x & 24 & 33 & 64 & 87 & & & 57.94 & 0.96 & 39.76 & 499.14\\
 & w & 26 & 38 & 20 & 16 & & & &&&\\
 \hline
 $w_3^*$ & x & 22 & 23 & 24 & 33 & 63 & 87 & 57.46 & 0.95 & 39.75 & 499.99\\
 & w & 1 & 2 & 24 & 39 & 19 & 15 & & & & \\
 \hline
$w_4^*$ & x & 0 & 14 & 24 & 34 & 64 & 87 & 56.75 & 0.94 & 39.47 & 499.86\\
 & w & 1 & 1 & 25 & 39 & 18 & 16 & & & & \\
 \hline
 $w_5^*$ & x & 23 & 33 & 43 & 55 & 65 & 86 & 53.45 & 0.89 & 36.94 & 499.70\\
 & w & 25 & 25 & 10 & 11 & 15 & 14 & & & & \\
 \hline
\end{tabular}
\caption{$D$-optimal designs for the model \eqref{CR} with sequentially appended constraints: size constraint $N=100$ ($w_0^*$, the first panel in Figure \ref{Fig:designs_all_gg}), expected number of failed trials at most 40 ($w_1^*$, the second panel in Figure \ref{Fig:designs_all_gg}), total cost at most $500$ ($w_2^*$, the third panel in Figure \ref{Fig:designs_all_gg}), minimum support size $6$ ($w_3^*$, the fourth panel in Figure \ref{Fig:designs_all_gg}), space between the doses at least $10$ ($w_4^*$, the fifth panel in Figure \ref{Fig:designs_all_gg}), number of replications between $10$ and $25$ ($w_5^*$, the last panel in Figure \ref{Fig:designs_all_gg}). See the text for further details on the constraints.}\label{Table:designs}
\end{table}

\section*{Discussion}

The computational approach introduced in this paper broadens the applicability of mathematical programming techniques originally developed for computing exact optimal designs in univariate-response regression models with standard linear constraints. By introducing a straightforward transformation strategy, we enable these methods to handle exact optimal design problems for multivariate-response models subject to both general linear and, most importantly, novel sparsity constraints. These sparsity constraints include, for instance, bounds on the number of distinct experimental conditions, support-influenced costs, or lower and upper replication limits at design points.

While our demonstrations focused on the D-optimality criterion and employed MISOCP solvers, the method is general and can accommodate any criterion for which a suitable mathematical programming formulation is (or will be) available. The flexibility and generality of this framework make it relevant to a broad range of applications, including dose-finding studies in clinical trials, as demonstrated in our numerical example.

Future methodological and algorithmic research directions include extending this approach to Bayesian, adaptive, or sequential design settings. It would also be valuable to demonstrate the utility of the proposed approach in other important application areas, such as screening experiments, which often require finer control over replications, and product lifetime testing, which naturally leads to multivariate and strongly correlated observations.

\bibliographystyle{plain}
\bibliography{DOSE}

\end{document}